\documentclass[journal,12pt,draftcls,onecolumn]{IEEEtran}

\usepackage[cmex10]{amsmath}
\usepackage{amssymb}
\usepackage{epsfig}
\usepackage{subfigure}
\usepackage{graphicx}

\def\BibTeX{{\rm B\kern-.05em{\sc i\kern-.025em b}\kern-.08em
    T\kern-.1667em\lower.7ex\hbox{E}\kern-.125emX}}

\newtheorem{theorem}{Theorem}

\newtheorem{lemma}{Lemma}

\newtheorem{remarks}{Remark}

\title{On the Existence of Optimal Exact-Repair MDS Codes for Distributed Storage}

\author{Changho Suh and Kannan Ramchandran \\
Wireless Foundations\\
University of California at Berkeley \\
Email: \{chsuh, kannanr\}@eecs.berkeley.edu}

\begin{document}

\IEEEaftertitletext{
\begin{abstract}
The high repair cost of $(n,k)$ Maximum Distance Separable (MDS) erasure codes has recently motivated a new class of codes, called Regenerating Codes, that optimally trade off storage cost for repair bandwidth.
In this paper, we address bandwidth-optimal $(n,k,d)$ Exact-Repair MDS codes, which allow for any failed node to be repaired exactly with access to arbitrary $d$ survivor nodes, where $k \leq d \leq n-1$. We show the existence of Exact-Repair MDS codes that achieve minimum repair bandwidth (matching the cutset lower bound) for arbitrary admissible $(n,k,d)$, i.e., $k < n$ and $k \leq d \leq n-1$. Our approach is based on interference alignment techniques and uses \emph{vector} linear codes which allow to split symbols into arbitrarily small subsymbols.

\end{abstract}
}

\maketitle

\section{Introduction}

In distributed storage systems, maximum distance separable (MDS) erasure codes are well-known coding schemes that can offer maximum reliability for a given storage overhead. For an  $(n,k)$ MDS code for storage, a source file of size $\mathcal{M}$ bits is divided equally into $k$ units (of size $\frac{\mathcal{M}}{k}$ bits each), and these $k$ data units are expanded into $n$ encoded units, and stored at $n$ nodes. The code guarantees that a user or Data Collector (DC) can reconstruct the source file by connecting to any arbitrary $k$ nodes. In other words, any $(n-k)$ node failures can be tolerated with a minimum storage cost of  $\frac{\mathcal{M}}{k}$ at each of $n$ nodes. While MDS codes are optimal in terms of reliability versus storage overhead, they come with a significant maintenance overhead when it comes to repairing failed encoded nodes to restore the MDS system-wide property.
Specifically, consider failure of a single encoded node and the cost needed to restore this node. It can be shown that this repair incurs an aggregate cost of $\mathcal{M}$ bits of information from $k$ nodes. Since each encoded unit contains only $\frac{\mathcal{M}}{k}$ bits of information, this represents a $k$-fold inefficiency with respect to the repair bandwidth.

This challenge has motivated a new class of coding schemes, called Regenerating Codes \cite{Dimakis:INFOCOM, Wu:Allerton}, which target the information-theoretic optimal tradeoff between storage cost and repair bandwidth. On one end of this spectrum of Regenerating Codes are Minimum Storage Regenerating (MSR) repair codes that can match the minimum storage cost of MDS codes while also significantly reducing repair bandwidth. As shown in \cite{Dimakis:INFOCOM, Wu:Allerton}, the fundamental tradeoff between bandwidth and storage depends on the number of nodes that are connected to repair a failed node, simply called the degee $d$ where $k \leq d \leq n-1$. The optimal tradeoff is characterized by
\begin{align}
\label{eq-MSRpoint}
 ( \alpha, \gamma) = \left( \frac{\mathcal{M}}{k}, \frac{\mathcal{M}}{k} \cdot \frac{d}{d-k+1} \right),
\end{align}
where $\alpha$ and $\gamma$ denote the optimal storage cost and repair bandwidth, respectively for repairing a single failed node, while retaining the MDS-code property for the user.
Note that this code requires the same minimal storage cost (of size $\frac{\mathcal{M}}{k}$) as that of conventional MDS codes, while substantially reducing repair bandwidth by a factor of $\frac{k(d-k+1)}{d}$ (e.g., for $(n,k,d)=(31,6,30)$, there is a $5$x bandwidth reduction).
MSR $(n,k,d)$ repair codes can be considered as Repair MDS codes that $(a)$ have an $(n,k)$ MDS-code property; and $(b)$ can repair single-node failures with minimum repair bandwidth given a repair-degree of $d$. Throughout this paper, we will use Repair MDS codes to indicate MSR repair codes.
%
%
%



While Repair MDS codes enjoy substantial benefits over conventional MDS codes, they come with some limitations in construction. Specifically, the achievable schemes in \cite{Dimakis:INFOCOM, Wu:Allerton} that meet the optimal tradeoff bound of (\ref{eq-MSRpoint}) restore failed nodes in a \emph{functional} manner only, using a random-network-coding based framework. This means that the replacement nodes maintain the MDS-code property (that any $k$ out of $n$ nodes can allow for the data to be reconstructed) but do not \emph{exactly} replicate the information content of the failed nodes.

Mere functional repair can be limiting. First, in many applications of interest, there is a need to maintain the code in systematic form, i.e., where  the user data in the form of $k$ information units are exactly stored at $k$ nodes and parity information (mixtures of $k$ information units) are stored at the remaining $(n-k)$ nodes. Secondly, under functional repair, additional overhead information needs to be exchanged for continually updating repairing-and-decoding rules whenever a failure occurs. This can significantly increase system overhead. A third problem is that the random-network-coding based solution of \cite{Dimakis:INFOCOM} can require a huge finite-field size, which can significantly increase the computational complexity of encoding-and-decoding\footnote{In \cite{Dimakis:INFOCOM}, Dimakis-Godfrey-Wu-Wainwright-Ramchandran translated the regenerating-codes problem into a multicast communication problem where random-network-coding-based schemes require a huge field size especially for large networks. In storage problems, the field size issue is further aggravated by the need to support a dynamically expanding network size due to the need for continual repair.}. Lastly, functional repair is undesirable in storage security applications in the face of eavesdroppers. In this case, information leakage occurs continually due to the dynamics of repairing-and-decoding rules that can be potentially observed by eavesdroppers \cite{Sameer:ISIT2010}.

These drawbacks motivate the need for \emph{exact} repair of failed nodes. This leads to the following question: is there a price for attaining the optimal tradeoff of (\ref{eq-MSRpoint}) with the extra constraint of exact repair? The work in~\cite{KumarRamchandran_MSR} considers partial exact repair (where only systematic nodes are repaired exactly), while the work in~\cite{SuhR_MSR} considers exact repair of all nodes, giving a clear answer with deterministic \emph{scalar linear} codes\footnote{In scalar linear codes, symbols are not allowed to be split into arbitrarily small subsymbols as with vector linear codes.} having small alphabet size for the case of $\frac{k}{n} \leq \frac{1}{2}$ (and $d \geq 2k-1$): it was shown that for this regime, there is no price even with the extra constraint of exact repair. What about for either $\frac{k}{n} > \frac{1}{2}$ or $k \leq d < 2k-1$? The work in~\cite{KumarRamchandran_MSR} sheds some light on this case: specifically, it was shown that under scalar linear codes, when either $\frac{k}{n} > \frac{1}{2} + \frac{2}{n}$ or $k+1 \leq d \leq \max(k+1, 2k-4)$, there \emph{is} a price for exact repair.
What if non-linear or \emph{vector} linear codes are used? The tightness of the optimal tradeoff of (\ref{eq-MSRpoint}) under these assumptions has remained open. In this paper, we show that using \emph{vector} linear codes, the optimal tradeoff of (\ref{eq-MSRpoint}) can be indeed attained \emph{for all admissible values of $(n,k,d)$}, i.e., $k <n$ and $k \leq d \leq n-1$. That is if we are willing to deal with arbitrarily small subsymbols, then Exact-Repair MDS codes can come with no loss of optimality over functional-repair MDS codes. Note that we will use this definition of admissibility throughout the paper.

Our achievable scheme builds on the concept of interference alignment, which was introduced in the context of wireless communication networks~\cite{Mohammad,Jafar:IC}.
In particular, the interference alignment scheme in~\cite{Jafar:IC} that permits an arbitrarily large number of symbol extensions (i.e., \emph{vector} linear codes) forms the basis of our results here. The results in~\cite{KumarRamchandran_MSR} say that under scalar linear codes, the case of either $\frac{k}{n} > \frac{1}{2} + \frac{2}{n}$ or $k+1 \leq d \leq \max(k+1, 2k-4)$ induces more constraints than the available number of design variables. This parallels the problem encountered by Cadambe and Jafar in~\cite{Jafar:IC} in the conceptually similar but physically different context of wireless interference channels. Cadambe and Jafar resolve this issue  in~\cite{Jafar:IC} using the idea of symbol-extension, which is analogous to the idea of vector linear codes for the distributed storage repair problem studied here. Building on the connection described in~\cite{SuhR_MSR} between the wireless interference and the distributed storage repair problems, we leverage the scheme introduced in~\cite{Jafar:IC} to the repair problem, showing the existence of Exact-Repair MDS codes that achieve minimum repair bandwidth (matching the cutset lower bound) for all admissible values of~$(n,k,d)$.

\section{Interference Alignment for Exact-Repair MDS Codes}
\label{sec-Notations}

Linear network coding \cite{Koetter:it,ahlswede:it} (that allows multiple messages to be linearly combined at network nodes) has been established recently as a useful tool for addressing interference issues even in wireline networks where all the communication links are orthogonal and non-interfering. This attribute was first observed in \cite{Wu:ISIT}, where it was shown that interference alignment could be exploited for storage networks, specifically for Exact-Repair MDS codes having small $k$ ($k=2$). However, generalizing interference alignment to large values of $k$ (even $k=3$) proves to be challenging, as we describe in the sequel. In order to appreciate this better, let us first review the scheme of~\cite{Wu:ISIT} that was applied to the exact repair problem. We will then address the difficulty of extending interference alignment for larger systems and describe how to address this in Section \ref{sec-BasisFramework}.

\subsection{Review of $(4,2)$ Exact-Repair MDS Codes~\cite{Wu:ISIT}}
Fig. \ref{fig:42example} illustrates an interference alignment scheme for a $(4,2)$ Exact-Repair MDS code defined over ${\sf GF}(5)$. First one can easily check the MDS property of the code, i.e., all the source files can be reconstructed from any $k(=2)$ nodes out of $n(=4)$ nodes. As an illustration, let us see how failed node 1 (storing $(a_1, a_2)$) can be exactly repaired. We assume that the degree $d$ is $3$, and a source file size  $\mathcal{M}$ is $4$. The cutset bound (\ref{eq-MSRpoint}) then gives the fundamental limits of: storage cost $\alpha = 2$; and repair-bandwidth-per-link $\beta:= \frac{\gamma}{d}=1$.

The example illustrated in Fig. \ref{fig:42example} shows that the parameter set described above is achievable using interference alignment. Here is a summary of the scheme. Recall that the bandwidth-per-link is $\beta=1$ and we use a scalar linear code, i.e., each symbol has unit capacity and cannot be split into arbitrarily small subsymbols. Hence, each survivor node uses a projection vector to project its data into a scalar. Choosing appropriate projection vectors, we get the equations as shown in Fig. \ref{fig:42example}: $(b_1 + b_2)$; $a_1 + 2a_2 + (b_1 + b_2)$; $2 a_1 + a_2 + (b_1 + b_2)$. Observe that the undesired signals $(b_1,b_2)$ (interference) are aligned onto an 1-dimensional linear subspace, thereby achieving interference alignment. Therefore, we can successfully decode $(a_1,a_2)$ with three equations although there are four unknowns.
\begin{figure}[t]
\begin{center}
{\epsfig{figure=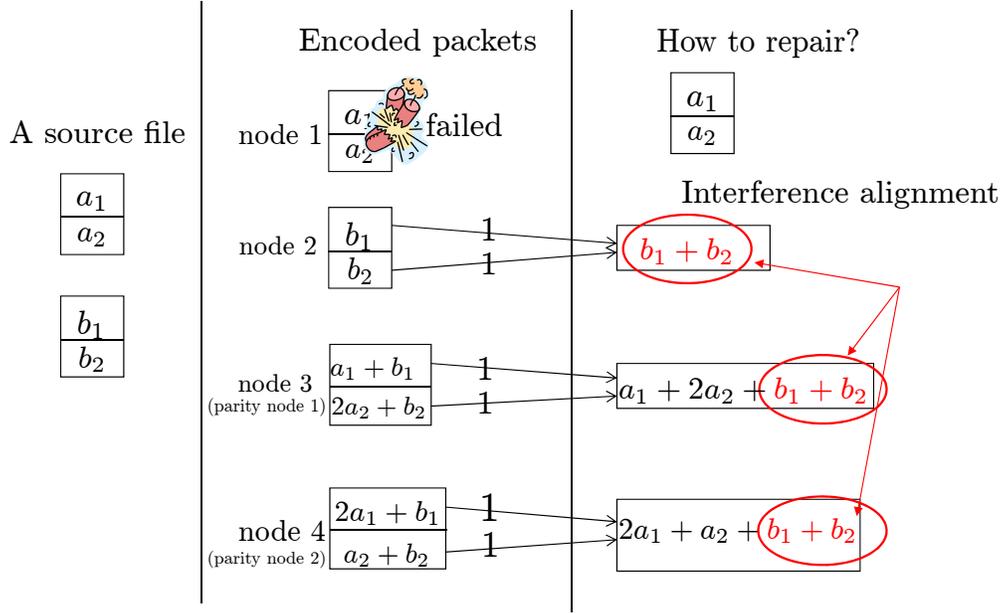, angle=0, width=0.8\textwidth}}
\end{center}
\caption{Interference alignment for a $(4,2)$ Exact-Repair MDS code defined over ${\sf GF}(5)$ \cite{Wu:ISIT}. Designing appropriate projection vectors, we can align interference space of $(b_1, b_2)$ into one-dimensional linear space spanned by $[1,\; 1]^t$. As a result, we can successfully decode 2 desired unknowns $(a_1,a_2)$ from 3 equations containing 4 unknowns $(a_1,a_2,b_1,b_2)$.} \label{fig:42example}

\end{figure}
Similarly, we can repair $(b_1,b_2)$ when it has failed.

For parity node repair, a remapping technique is introduced. The idea is to define parity node symbols with new variables as follows:
\begin{align*}
&\textrm{Node 3: }a_1' := a_1 + b_1; \;\;  a_2' := 2a_2 + b_2;\\
&\textrm{Node 4: }b_1' := 2a_1 + b_1;\;\;  b_2' :=a_2 + b_2.
\end{align*}
We can then rewrite $(a_1,a_2)$ and $(b_1,b_2)$ with respect to $(a_1',a_2')$ and $(b_1',b_2')$. In terms of prime notation, parity nodes turn into systematic nodes and vice versa. With this remapping, one can easily design projection vectors for exact repair of parity nodes.

\begin{figure}[t]
\begin{center}
{\epsfig{figure=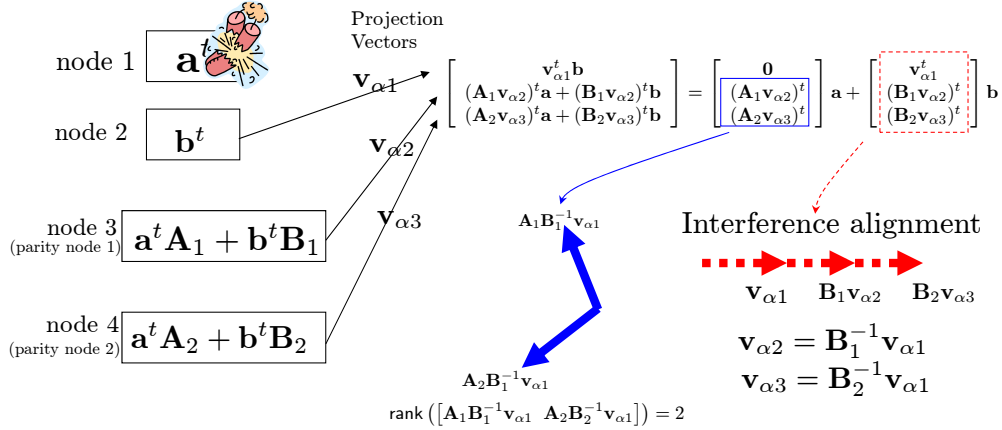, angle=0, width=0.8\textwidth}}
\end{center}
\caption{Geometric interpretation of interference alignment. The blue solid-line and red dashed-line vectors indicate linear subspaces with respect to ``$\mathbf{a}$'' and ``$\mathbf{b}$'', respectively. The choice of $\mathbf{v}_{\alpha 2} = \mathbf{B}_1^{-1} \mathbf{v}_{\alpha 1}$ and $\mathbf{v}_{\alpha 3} = \mathbf{B}_2^{-1} \mathbf{v}_{\alpha 1}$ enables interference alignment. For the specific example of Fig. \ref{fig:42example}, the corresponding encoding matrices are
$\mathbf{A}_1 = \left[1, 0; 0, 2 \right]$, $\mathbf{B}_1 = \left[1, 0; 0,1 \right]$.
$\mathbf{A}_2 = \left[ 2,  0; 0, 1 \right]$, $\mathbf{B}_2 = \left[ 1,  0; 0,  1 \right]$.} \label{fig:GeometricView}
\end{figure}

\subsection{Geometric Interpretation}
Using matrix notation, we provide geometric interpretation of interference alignment for the same example in Fig.~\ref{fig:42example}. Let $\mathbf{a}= (a_1,a_2)^t$ and $\mathbf{b}= (b_1,b_2)^t$  be 2-dimensional information-unit vectors, where $(\cdot)^t$ indicates a transpose. Let $\mathbf{A}_i$ and $\mathbf{B}_i$ be $2$-by-$2$ encoding submatrices for parity node $i$ ($i=1,2$). Finally we define 2-dimensional projection vectors $\mathbf{v}_{\alpha i}$'s ($i=1,2,3$).

Let us consider exact repair of systematic node 1. By connecting to three nodes, we get: $\mathbf{b}^{t} \mathbf{v}_{\alpha 1}$; $\mathbf{a}^t (\mathbf{A}_1 \mathbf{v}_{\alpha 2}) +\mathbf{b}^t (\mathbf{B}_1  \mathbf{v}_{\alpha 2})$;
$ \mathbf{a}^t ( \mathbf{A}_2 \mathbf{v}_{\alpha 3} ) + \mathbf{b}^t ( \mathbf{B}_2  \mathbf{v}_{\alpha 3} )$. Recall the goal of decoding 2 desired unknowns out of 3 equations including 4 unknowns. To achieve this goal, we need:
\begin{align}
{\sf rank} \left( \left[
    \begin{array}{c}
      (\mathbf{A}_1 \mathbf{v}_{\alpha 2})^t \\
      ( \mathbf{A}_2 \mathbf{v}_{\alpha 3} )^t \\
    \end{array}
  \right] \right) = 2; \;\; {\sf rank} \left( \left[
    \begin{array}{c}
      \mathbf{v}_{\alpha 1}^t \\
      (\mathbf{B}_1  \mathbf{v}_{\alpha 2})^t \\
      ( \mathbf{B}_2  \mathbf{v}_{\alpha 3} )^t \\
    \end{array}
  \right] \right) = 1.
\end{align}
The second condition can be met by setting $\mathbf{v}_{\alpha 2} = \mathbf{B}_1^{-1} \mathbf{v}_{\alpha 1}$ and $\mathbf{v}_{\alpha 3} = \mathbf{B}_2^{-1} \mathbf{v}_{\alpha 1}$. This choice forces the interference space to be collapsed into a one-dimensional linear subspace, thereby achieving interference alignment. With this setting, the first condition now becomes
\begin{align}
\label{eq_42_1}
\mathsf{rank} \left( \left[ \mathbf{A}_1 \mathbf{B}_1^{-1} \mathbf{v}_{\alpha 1} \;\;
\mathbf{A}_2  \mathbf{B}_2^{-1} \mathbf{v}_{\alpha 1} \right] \right) = 2.
\end{align}
It can be easily verified that the choice of $\mathbf{A}_i$'s and $\mathbf{B}_i$'s given in Figs. \ref{fig:42example} and \ref{fig:GeometricView} guarantees the above condition.
When the node 2 fails, we get a similar condition:
\begin{align}
\label{eq_42_2}
\mathsf{rank} \left( \left[ \mathbf{B}_1 \mathbf{A}_1^{-1} \mathbf{v}_{\beta 1} \;\;
\mathbf{B}_2  \mathbf{A}_2^{-1} \mathbf{v}_{\beta 1} \right] \right) = 2,
\end{align}
where $\mathbf{v}_{\beta i}$'s denote projection vectors for node 2 repair. This condition also holds under the given choice of encoding matrices.
With this remapping, one can easily design projection vectors for exact repair of parity nodes.

\subsection{Connection with Interference Channels in Communication Problems}
Observe the three equations shown in Fig. \ref{fig:GeometricView}:
\begin{align*}
\underbrace{\left[
    \begin{array}{c}
      \mathbf{0} \\
      (\mathbf{A}_1 \mathbf{v}_{\alpha 2})^t \\
      ( \mathbf{A}_2 \mathbf{v}_{\alpha 3} )^t \\
    \end{array}
  \right] \mathbf{a}}_{desired\;signals} + \underbrace{\left[
    \begin{array}{cc}
      \mathbf{v}_{\alpha 1}^t \\
      (\mathbf{B}_1  \mathbf{v}_{\alpha 2})^t \\
      ( \mathbf{B}_2  \mathbf{v}_{\alpha 3} )^t \\
    \end{array}
  \right] \mathbf{b}}_{interference}.
\end{align*}
Separating into two parts, we can view this problem as a wireless communication problem, wherein a subset of the information is desired to be decoded in the presence of interference. Note that for each term (e.g., $\mathbf{A}_1 \mathbf{v}_{\alpha 2}$), the matrix $\mathbf{A}_1$ and vector $\mathbf{v}_{\alpha 2}$ correspond to channel matrix and transmission vector in wireless communication problems, respectively.

There are, however, significant differences. In the wireless communication problem, the channel matrices are provided by nature and therefore not controllable. The transmission strategy alone (vector variables) can be controlled for achieving interference alignment. On the other hand, in our storage repair problems, both matrices and vectors are controllable, i.e., projection vectors and encoding matrices can be arbitrarily designed, resulting in more flexibility. However, our storage repair problem comes with unparalleled challenges due to the MDS requirement and the  multiple failure configurations. These induce multiple interference alignment constraints that need to be simultaneously satisfied.
What makes this difficult is that the encoding matrices, once designed, must be the same for all repair configurations. This is particularly acute for large values of $k$ (even $k=3$), as the number of possible failure configurations increases with $n$ (which increases with $k$).


\section{A Proposed Framework for Exact-Repair MDS Codes}
\label{sec-BasisFramework}
%

We propose a conceptual framework based on vector linear codes to address the exact repair problem. As described earlier, this framework is based on that of interference alignment for wireless channels in~\cite{Jafar:IC}.
We leverage the connection between the two problems to develop Exact-Repair MDS codes that are optimal in repair bandwidth for all admissible values of $(n,k,d)$.


Our framework consists of four components:
(1) developing a code structure for exact repair of systematic nodes based on the vector linear codes; (2) drawing a \emph{dual} structure between the systematic and parity node repair; (3) guaranteeing the MDS-code property; (4) providing a probabilistic guarantee of the existence of the code for a large enough alphabet size. In particular, the
diagonal structure of single-antenna wireless channels (exploited in~\cite{Jafar:IC}) forms the basis of the structure of encoding submatrices of our codes. The framework covers all admissible values of $(n,k,d)$. This contrasts the scalar-linear code based framework in~\cite{SuhR_MSR} which covers the case of $\frac{k}{n} \leq \frac{1}{2}$ and $d \geq 2k-1$, but which provides deterministic codes with small alphabet size and guaranteed zero error.
Furthermore, addressing different code parameters in the case of $\frac{k}{n} \leq \frac{1}{2}$ and $d \geq 2k-1$ requires specific attention, such as the design of puncturing codes introduced in~\cite{KumarRamchandran_MSR}. See~\cite{SuhR_MSR} for details. In contrast, here we target only the existence of exact-repair codes without specifying constructions. This allows for a simpler characterization of the solution space for the entire range of admissible repair code parameters. In order to convey the concepts in a clear and concise manner, we first focus on the simplest example which does not belong to the framework in~\cite{SuhR_MSR}: $(6,3,4)$ Exact-Repair MDS codes. This example is a representative of the general case of $k<n$ and $k \leq d \leq n-1$, with the generalization following in a straightforward way from this example. This will be discussed in Section~\ref{sec-generalization}.


\subsection{Systematic Node Repair}
\label{sec:FrameworkSystematic}

For $k \geq 3$ (more-than-two interfering information units), achieving interference alignment for exact repair turns out to be significantly more complex than the $k=2$ case. Fig.~\ref{fig:63_EMSR_Challenge} illustrates this difficulty through the example of repairing node 1 for a $(6,3,4)$ code. In accordance with the $(4,2)$ code example in Figs.~\ref{fig:42example} and~\ref{fig:GeometricView}, we choose $\mathcal{M}=6$ so that repair-bandwidth-per-link has unit capacity ($\beta:= \frac{\gamma}{d}=1$). By the optimal tradeoff of (\ref{eq-MSRpoint}), this gives $\alpha=2$. Suppose that we use scalar linear codes, i.e., each symbol has unit capacity and cannot be chopped up into arbitrarily smaller chunks. We define $\mathbf{a}=(a_1,a_2)^t$, $\mathbf{b}=(b_1,b_2)^t$ and $\mathbf{c}=(c_1,c_2)^t$. We define 2-by-2 encoding submatrices of $\mathbf{A}_i$, $\mathbf{B}_i$ and  $\mathbf{C}_i$ (for $i=1,2,3$); and 2-dimensional projection vectors $\mathbf{v}_{\alpha i}$'s.

\begin{figure}[t]
\begin{center}
{\epsfig{figure=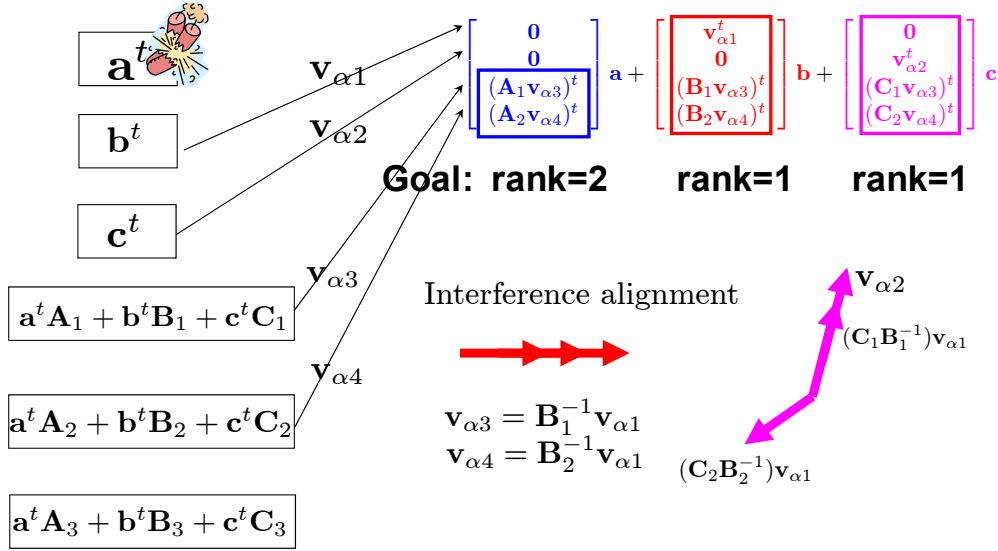, angle=0, width=0.8\textwidth}}
\end{center}
\caption{Difficulty of achieving interference alignment simultaneously when using scalar linear codes}\label{fig:63_EMSR_Challenge}
\end{figure}

Suppose that survivor nodes $(2,3,4,5)$ are connected for exact repair of node 1. We then get the 4 $(=d)$ equations:
\begin{align*}
  \left[
    \begin{array}{c}
      \mathbf{0} \\
      \mathbf{0}  \\
      (\mathbf{A}_1 \mathbf{v}_{\alpha 3})^t \\
       ( \mathbf{A}_2 \mathbf{v}_{\alpha 4} )^t  \\
     \end{array}
  \right] \mathbf{a} +
 \left[
    \begin{array}{c}
      \mathbf{v}_{\alpha 1}^t  \\
       \mathbf{0} \\
        (\mathbf{B}_1  \mathbf{v}_{\alpha 3})^t \\
       ( \mathbf{B}_2  \mathbf{v}_{\alpha 4} )^t \\
         \end{array}
  \right] \mathbf{b} + \left[
    \begin{array}{c}
       \mathbf{0} \\
        \mathbf{v}_{\alpha 2}^t \\
     (\mathbf{C}_1  \mathbf{v}_{\alpha 3})^t \\
        (\mathbf{C}_2  \mathbf{v}_{\alpha 4})^t \\
     \end{array}
  \right] \mathbf{c}.
  \end{align*}
In order to successfully recover the desired signal components of ``$\mathbf{a}$'', the matrices associated with $\mathbf{b}$ and $\mathbf{c}$ should have rank 1, respectively, while the matrix associated with $\mathbf{a}$ should have full rank of 3. In accordance with the $(4,2)$ code example in Fig.~\ref{fig:GeometricView}, if one were to set $\mathbf{v}_{\alpha 3} = \mathbf{B}_1^{-1} \mathbf{v}_{\alpha 1}$ and $\mathbf{v}_{\alpha 4} = \mathbf{B}_2^{-1} \mathbf{v}_{\alpha 1}$, then it is possible to achieve interference alignment with respect to $\mathbf{b}$. However, this choice also specifies the interference space of $\mathbf{c}$. If the $\mathbf{B}_i$'s and $\mathbf{C}_i$'s are not designed judiciously, interference alignment is not guaranteed for $\mathbf{c}$. Hence, it is not evident how to achieve interference alignment at the same time.

\begin{figure}[t]
\begin{center}
{\epsfig{figure=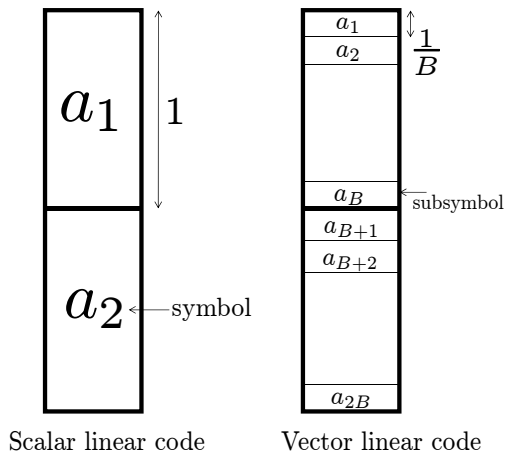, angle=0, width=0.4\textwidth}}
\end{center}
\caption{Illustration of the idea of vector linear codes through storage node 1 in the $(6,3,4)$ code example. In scalar linear codes, symbols are not allowed to be split. On the other hand, vector linear codes allow to split symbols into arbitrarily small subsymbols. In this example, node 1 stores $\alpha=2$ symbols, each of which has unit capacity. In vector linear codes, this unit-capacity symbol can be split into subsymbols with arbitrarily small capacity. For example, we can split each symbol into $B$ number of subsymbols, so each subsymbol has $\frac{1}{B}$ capacity.}\label{fig:VectorLinearCodes}
\end{figure}

In order to address the challenge of simultaneous interference alignment, we invoke the idea of symbol extension introduced in~\cite{Jafar:IC}, which is equivalent to the concept of vector linear codes in the storage repair problem. Fig.~\ref{fig:VectorLinearCodes} illustrates the idea of vector linear codes through storage node 1 in the $(6,3,4)$ code example. While scalar linear codes do not allow symbol splitting, vector linear codes permit the splitting of symbols into arbitrarily small subsymbols. In this example, each node stores $\alpha=2$ symbols, each of which has unit capacity. In vector linear codes, this unit-capacity symbol is allowed to be split into subsymbols with arbitrary small capacity. In this example, we split each symbol into $B$ number of subsymbols, so each subsymbol has $\frac{1}{B}$ capacity.

\begin{figure}[t]
\begin{center}
{\epsfig{figure=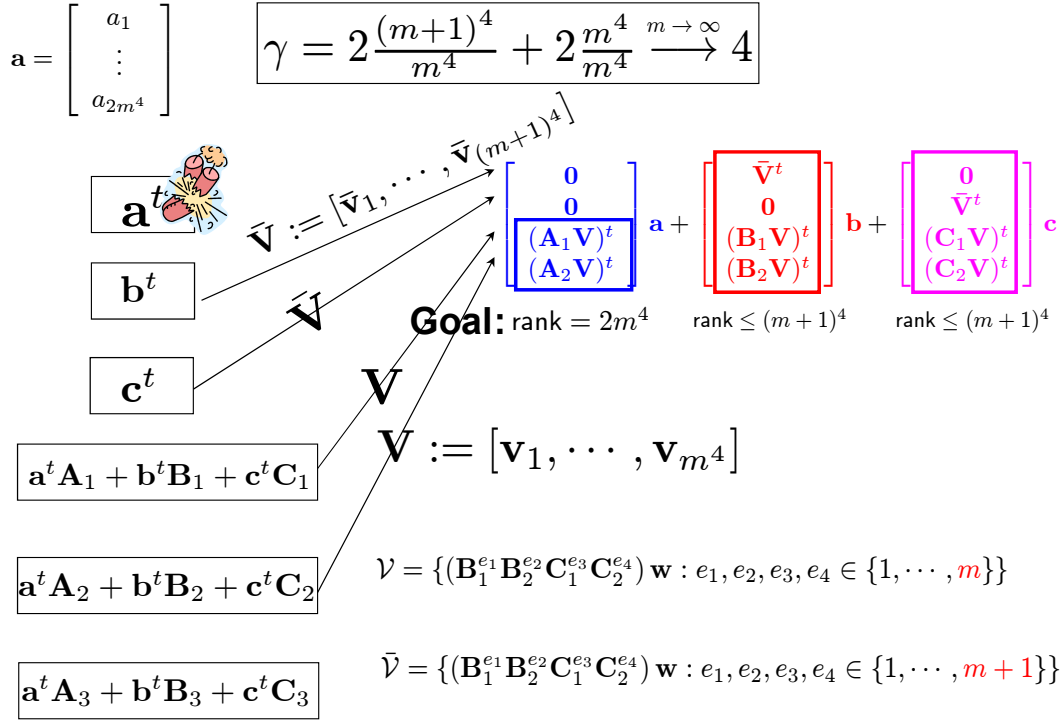, angle=0, width=0.85\textwidth}}
\end{center}
\caption{Illustration of exact repair of systematic node 1 for $(6,3,4)$ Exact-Repair MDS codes. We split each symbol into $B=m^N$ number of subsymbols, where $m$ is an arbitrarily large positive integer and the exponent $N$ is equal to 4 and is carefully chosen depending on code parameters, i.e., $N=(k-1)(d-k+1)=4$. This corresponds to the total number of encoding submatrices involved in the connection except for those associated with desired signals. Note that each subsymbol has $\frac{1}{m^4}$ capacity. The maximum file size (based on the optimal tradeoff of~(\ref{eq-MSRpoint})) is $\mathcal{M}=6 \; \textrm{units}$, inducing a storage cost $\alpha=2$ units. Hence, each storage contains $2m^4$ number of subsymbols and the size of encoding submatrices is $2m^4$-by-$2m^4$. We consider diagonal encoding submatrices. A failed node is exactly repaired by having systematic and parity survivor nodes project their data onto linear subspaces spanned by column vectors of $\mathbf{\bar{V}}:=[\mathbf{\bar{v}}_1, \cdots, \mathbf{\bar{v}}_{(m+1)^4}]$ and $\mathbf{V}:=[\mathbf{v}_1, \cdots, \mathbf{v}_{m^4}]$, respectively. Here $\mathbf{\bar{v}}_i \in \mathcal{\bar{V}}$ and $\mathbf{v}_i \in \mathcal{V}$. Notice that $\mathbf{B}_1 \mathbf{v}_i, \mathbf{B}_2 \mathbf{v}_i, \mathbf{C}_1 \mathbf{v}_i, \mathbf{C}_2 \mathbf{v}_i \in \mathcal{\bar{V}}, \forall i=1,\cdots, m^4$. Hence, the matrix associated with interference $\mathbf{b}$ has rank of at most $(m+1)^4$ instead of $2m^4$. Similarly the matrix associated with interference $\mathbf{c}$ has rank of at most $(m+1)^4$. This enables simultaneous interference alignment as $m \rightarrow \infty$. On the other hand, ${\sf rank}[\mathbf{A}_1 \mathbf{V}, \mathbf{A}_2 \mathbf{V}]= 2 m^4$ with probability 1, providing probabilistic guarantee of decodability of desired signals. Finally, notice that total repair bandwidth $\gamma = 2 \frac{(m+1)^4}{m^4} + 2 \cdot 1$ approaches the cutset lower bound of 4 units as $m$ goes to infinity. Therefore, we can ensure exact repair of systematic node 1 with minimum repair bandwidth matching the cutset lower bound.}\label{fig:63_EMSR}
\end{figure}

This idea of vector linear codes is key to interference alignment for the storage repair problem. Fig.~\ref{fig:63_EMSR} illustrates exact repair of systematic node 1 for $(6,3,4)$ Exact-Repair MDS codes. Using vector linear codes, we split each symbol into $B=m^N$ number of subsymbols, where $m$ is an arbitrarily large positive integer and the exponent $N$ is carefully chosen depending on code parameters. Specifically,
\begin{align}
N = (k-1) (d-k+1).
\end{align}
This choice of $N$ and the form of $B=m^N$ are closely related to the scheme to be described in the sequel. In this example, $N=4$. The maximum file size (based on the cutset bound of~(\ref{eq-MSRpoint})) is $\mathcal{M}=6$ units, inducing a storage cost $\alpha=2$ units. Since each subsymbol has $\frac{1}{m^4}$ capacity, each storage contains $\alpha m^4 (= 2 m^4)$ number of subsymbols, e.g., $\mathbf{a}^t = (a_1, \cdots, a_{2m^4}),$ where $a_i$ indicates a subsymbol. Note that the size of encoding submatrices ($\mathbf{A}_i, \mathbf{B}_i, \mathbf{C}_i$) is $2m^4$-by-$2m^4$. We consider \emph{diagonal} encoding submatrices. As pointed out in~\cite{Jafar:IC}, the diagonal matrix structure ensures \emph{commutativity} and this property provides the key to the interference alignment scheme (to be described shortly):
\begin{align}
\mathbf{A}_i = \left[
                 \begin{array}{cccc}
                   \alpha_{i,1} & 0 & \cdots & 0 \\
                   0 &  \alpha_{i,2}&  \cdots & 0 \\
                   \vdots & \vdots & \ddots & \vdots \\
                   0 & \cdots & 0 & \alpha_{i,2m^4} \\
                 \end{array}
               \right] (\textrm{\emph{commutative} property holds}).
\end{align}

A failed node 1 is exactly repaired through the following steps. Suppose without loss of generality that survivor nodes $(2,3,4,5)$ are used for exact repair of node 1, i.e., $k-1=2$ systematic nodes and $d-k+1=2$ parity nodes. One can alternatively use 1 systematic node and 3 parity nodes for repair instead. This does not fundamentally alter the analysis, and will be covered in Remark~\ref{remark:Arbitaryd} in the next section. For the time being, assume the above configuration for the connection: $k-1$ systematic nodes and $d-k+1$ parity nodes. Parity survivor nodes project their data using the following \emph{projection matrix}:
\begin{align}
\mathbf{V}:=[\mathbf{v}_1, \cdots, \mathbf{v}_{m^4}] \in \mathbb{F}_q^{2m^4 \times m^4},
\end{align}
where $\mathbf{v}_i \in \mathcal{V}$. The set $\mathcal{V}$ is defined as:
\begin{align}
\label{eq:Vform}
\mathcal{V} := \left \{ \left( \mathbf{B}_1^{e_1}
\mathbf{B}_2^{e_2}
\mathbf{C}_1^{e_3}
\mathbf{C}_2^{e_4}
\right) \mathbf{w}: e_1, e_2, e_3, e_4  \in \left \{ 1, \cdots , m \right \} \right \},
\end{align}
where $\mathbf{w}=[1,\cdots,1]^t$. Note that $|\mathcal{V}| \leq m^4$. The vector $\mathbf{v}_i$ maps to a different sequence of $(e_1,e_2,e_3,e_4)$. For example, we can map:
\begin{align}
\begin{split}
\label{eq:vimapping}
\mathbf{v}_1 &= \mathbf{B}_1
\mathbf{B}_2
\mathbf{C}_1
\mathbf{C}_2 \mathbf{w}, \\
\mathbf{v}_2 &= \mathbf{B}_1^2
\mathbf{B}_2
\mathbf{C}_1
\mathbf{C}_2 \mathbf{w}, \\
\mathbf{v}_3 &= \mathbf{B}_1^3
\mathbf{B}_2
\mathbf{C}_1
\mathbf{C}_2\mathbf{w} , \\
&\vdots \\
\mathbf{v}_{m^4-1} &= \mathbf{B}_1^m
\mathbf{B}_2^m
\mathbf{C}_1^m
\mathbf{C}_2^{m-1} \mathbf{w}, \\
\mathbf{v}_{m^4} &= \mathbf{B}_1^m
\mathbf{B}_2^m
\mathbf{C}_1^m
\mathbf{C}_2^m \mathbf{w}.
\end{split}
\end{align}
Let us consider the equations downloaded from parity node 1 and 2 (node 4 and 5):
\begin{align}
\begin{split}
\mathbf{a}^t (\mathbf{A}_1 \mathbf{V}) +
\mathbf{b}^t (\mathbf{B}_1 \mathbf{V}) +
\mathbf{c}^t (\mathbf{C}_1 \mathbf{V}); \\
\mathbf{a}^t (\mathbf{A}_2 \mathbf{V}) +
\mathbf{b}^t (\mathbf{B}_2 \mathbf{V}) +
\mathbf{c}^t (\mathbf{C}_2 \mathbf{V}).
\end{split}
\end{align}
Note that by (\ref{eq:Vform}), any column vector in $[\mathbf{B}_1 \mathbf{V}, \mathbf{B}_2 \mathbf{V}]$ or $[\mathbf{C}_1 \mathbf{V}, \mathbf{C}_2 \mathbf{V}]$ is an element of $\mathcal{\bar{V}}$ defined as:
\begin{align}
\label{eq:Vbarform}
\mathcal{\bar{V}} := \left \{ \left( \mathbf{B}_1^{e_1}
\mathbf{B}_2^{e_2}
\mathbf{C}_1^{e_3}
\mathbf{C}_2^{e_4}
\right) \mathbf{w}: e_1, e_2, e_3, e_4  \in \left \{ 1, \cdots , m +1 \right \} \right \}.
\end{align}
This implies that ${\sf rank}[\mathbf{B}_1 \mathbf{V}, \mathbf{B}_2 \mathbf{V}] \leq (m+1)^4$ and ${\sf rank}[\mathbf{C}_1 \mathbf{V}, \mathbf{C}_2 \mathbf{V}] \leq (m+1)^4$. This allows for simultaneous interference alignment. Systematic survivor nodes project their data using the following \emph{projection matrix}:
\begin{align}
\mathbf{\bar{V}}:=[\mathbf{\bar{v}}_1, \cdots, \mathbf{\bar{v}}_{(m+1)^4}] \in \mathbb{F}_q^{2m^4 \times (m+1)^4},
\end{align}
where $\mathbf{\bar{v}}_i \in \mathcal{\bar{V}}$. We also map $\mathbf{\bar{v}}_i$ to a difference sequence of $(e_1,e_2,e_3,e_4)$ as in~(\ref{eq:vimapping}). We can then guarantee that:
\begin{align}
\begin{split}
&{\sf span}[\mathbf{B}_1 \mathbf{V}, \mathbf{B}_2 \mathbf{V}] \subset {\sf span [\mathbf{\bar{V}}}]  \\
&{\sf span}[\mathbf{C}_1 \mathbf{V}, \mathbf{C}_2 \mathbf{V}] \subset {\sf span} [\mathbf{\bar{V}}].
\end{split}
\end{align}
Hence, we can completely get rid of any interference. Now let us analyze the decodability of the desired signal vector. To successfully recover $\mathbf{a}$, we need:
\begin{align}
\label{eq:fullrank_desired}
{\sf rank}[\mathbf{A}_1 \mathbf{V}, \mathbf{A}_2 \mathbf{V}] = 2m^4.
\end{align}
Using standard arguments based on the technique in~\cite{Jafar:IC} and Schwartz-Zippel lemma~\cite{Motwani}, we can ensure the condition of~(\ref{eq:fullrank_desired}) \emph{with probability 1} for a sufficiently large field size $q$. Specifically, we randomly and uniformly choose each diagonal element (non-zero) of all of the encoding submatrices in $\mathbb{F}_q$. We then compute the determinant of $[\mathbf{A}_1 \mathbf{V}, \mathbf{A}_2 \mathbf{V}]$ by adapting the technique in~\cite{Jafar:IC}. Using Schwartz-Zippel lemma~\cite{Motwani}, we can then show that the probability that the polynomial of the determinant is identically zero goes to zero for a sufficiently large field size. The proof is tedious and therefore we omit details. See~\cite{Jafar:IC,Motwani} for details.

We now validate that total repair bandwidth is:
\begin{align}
\begin{split}
\gamma &= (k-1) \frac{(m+1)^4}{m^4} + (d-k+1) \cdot \frac{m^4}{m^4} \\
&=  2 \frac{(m+1)^4}{m^4} + 2 \cdot 1  \\
&\longrightarrow 4 \textrm{ units}.
\end{split}
\end{align}
The first equality is because each subsymbol has capacity of $\frac{1}{m^4}$ and we use projection matrix $\mathbf{\bar{V}} \in \mathbb{F}_q^{2m^4 \times (m+1)^4}$ and $\mathbf{V} \in \mathbb{F}_q^{2m^4 \times m^4}$ when connecting to systematic nodes and parity nodes, respectively. Note that as $m$ goes to infinity, total repair bandwidth approaches minimum repair bandwidth matching the cutset lower bound of~(\ref{eq-MSRpoint}).
%


\subsection{Dual Relationship between Systematic and Parity Node Repair}
\label{sec:ParityNodeRepair}

We will show that parity nodes can be repaired by drawing a \emph{dual} relationship with systematic nodes. The procedure has two steps. The first is to remap parity nodes with $\mathbf{a}'$, $\mathbf{b}'$, and $\mathbf{c}'$, respectively:
\begin{align*}
\left[
    \begin{array}{ccc}
      \mathbf{a}'\\
      \mathbf{b}'\\
      \mathbf{c}'\\
    \end{array}
  \right]:=
\left[
    \begin{array}{ccc}
      \mathbf{A}_1^t & \mathbf{B}_1^t & \mathbf{C}_1^t\\
      \mathbf{A}_2^t & \mathbf{B}_2^t & \mathbf{C}_2^t\\
      \mathbf{A}_3^t & \mathbf{B}_3^t & \mathbf{C}_3^t\\
    \end{array}
  \right] \left[
    \begin{array}{ccc}
      \mathbf{a}\\
      \mathbf{b}\\
      \mathbf{c}\\
    \end{array}
  \right].
\end{align*}
Systematic nodes can then be rewritten in terms of the prime notations:
\begin{align}
\begin{split}
\mathbf{a}^t &= \mathbf{a}'^t \mathbf{A}_1' + \mathbf{b}'^t \mathbf{B}_1' +
\mathbf{c}'^t \mathbf{C}_1', \\
\mathbf{b}^t &= \mathbf{a}'^t \mathbf{A}_2' + \mathbf{b}'^t \mathbf{B}_2' +
\mathbf{c}'^t \mathbf{C}_2', \\
\mathbf{c}^t &= \mathbf{a}'^t \mathbf{A}_3' + \mathbf{b}'^t \mathbf{B}_3' +
\mathbf{c}'^t \mathbf{C}_3',
\end{split}
\end{align}
where the newly mapped encoding matrices $(\mathbf{A}_i', \mathbf{B}_i', \mathbf{C}_i)$'s are defined as:
\begin{align}
\label{eq-63_mapping}
\left[
   \begin{array}{ccc}
     \mathbf{A}_1'  & \mathbf{A}_2' & \mathbf{A}_3' \\
     \mathbf{B}_1'  & \mathbf{B}_2' & \mathbf{B}_3' \\
\mathbf{C}_1' & \mathbf{C}_2' & \mathbf{C}_3' \\
   \end{array}
 \right]: = \left[
   \begin{array}{ccc}
     \mathbf{A}_1  & \mathbf{A}_2 & \mathbf{A}_3 \\
     \mathbf{B}_1  & \mathbf{B}_2 & \mathbf{B}_3 \\
\mathbf{C}_1 & \mathbf{C}_2 & \mathbf{C}_3 \\
   \end{array}
 \right]^{-1}.
\end{align}
As in Section~\ref{sec:FrameworkSystematic}, we consider random construction of the code, i.e., each diagonal element in each encoding submatrix is i.i.d. uniformly drawn from $\mathbb{F}_q \setminus \{ 0 \}$.  Then, for a sufficiently large field size, the above composite matrix has non-zero determinant with probability 1 (again due to Schwartz-Zippel lemma). With this remapping, one can now dualize the relationship between systematic and parity node repair. Specifically, if all of the $\mathbf{A}_i'$'s, $\mathbf{B}_i'$'s, and $\mathbf{C}_i'$'s are \emph{diagonal} matrices, then exact repair of the parity nodes becomes transparent, as illustrated in Fig.~\ref{fig:63_EMSR_parity}.
Indeed $\mathbf{A}_i'$'s, $\mathbf{B}_i'$'s, and $\mathbf{C}_i'$'s are diagonal matrices, since these matrices are functions of diagonal matrices $\mathbf{A}_i$'s, $\mathbf{B}_i$'s, and $\mathbf{C}_i$'s. Therefore, following the same procedure in Section~\ref{sec:FrameworkSystematic}, we can guarantee exact repair of parity nodes with probability 1.


\begin{figure}[t]
\begin{center}
{\epsfig{figure=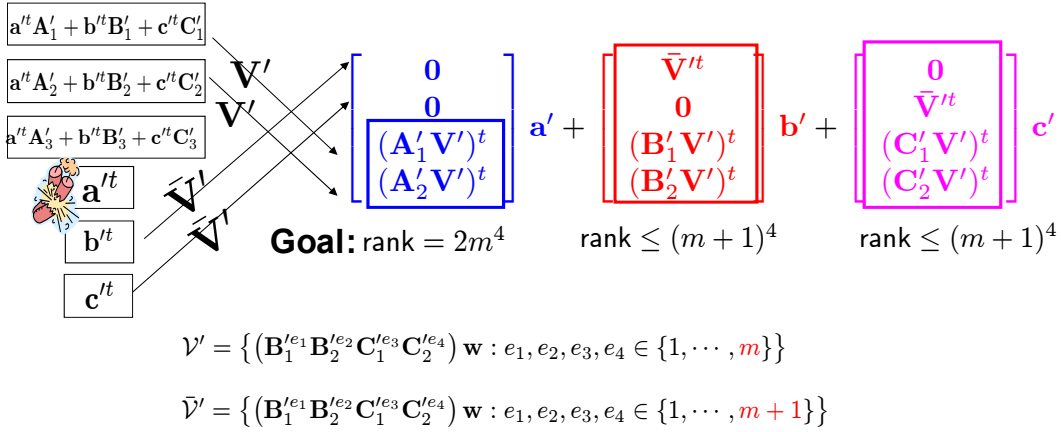, angle=0, width=0.85\textwidth}}
\end{center}
\caption{Illustration of exact repair of parity node 1 for $(6,3,4)$ Exact-Repair MDS codes. Notice that $\mathbf{A}_i'$'s, $\mathbf{B}_i'$'s, and $\mathbf{C}_i'$'s are also diagonal matrices, since these matrices are functions of diagonal matrices $\mathbf{A}_i$'s, $\mathbf{B}_i$'s, and $\mathbf{C}_i$'s. Survivor nodes 1 and 2 project their data onto linear subspaces spanned by column vectors of $\mathbf{V}':=[\mathbf{v}_1', \cdots, \mathbf{v}_{m^4}']$. Here $\mathbf{v}_i' \in \mathcal{V}'$. Notice that $\mathbf{B}_1' \mathbf{v}_i', \mathbf{B}_2' \mathbf{v}_i', \mathbf{C}_1' \mathbf{v}_i', \mathbf{C}_2' \mathbf{v}_i' \in \mathcal{\bar{V}}, \forall i=1,\cdots, m^4$. Hence, the matrix associated with interference $\mathbf{b}'$ has rank of at most $(m+1)^4$ instead of $2m^4$. Similarly the matrix associated with interference $\mathbf{c}'$ has rank of at most $(m+1)^4$. This enables simultaneous interference alignment as $m \rightarrow \infty$. Survival nodes 5 and 6 project their data using $\mathbf{\bar{V}}:=[\mathbf{\bar{v}}_1, \cdots, \mathbf{\bar{v}}_{(m+1)^4}]$ where $\mathbf{\bar{v}}_i \in \mathcal{\bar{V}}$. We can then clean out any interference. On the other hand, it is guaranteed that ${\sf rank}[\mathbf{A}_1' \mathbf{V}', \mathbf{A}_2' \mathbf{V}']= 2 m^4$ with probability 1, guaranteeing of decodability of desired signals with probability 1.}\label{fig:63_EMSR_parity}
\end{figure}

\begin{remarks}[Connecting to arbitrary $d$ nodes suffice for exact repair]
\label{remark:Arbitaryd}
In Section~\ref{sec:FrameworkSystematic}, we considered the only one connection configuration for exact repair: connecting to $k-1$ systematic nodes and $d-k+1$ parity nodes. We now consider other connection configurations. For example, consider the case when node 1 fails, as shown in Fig.~\ref{fig:63_EMSR}. Suppose we connect to nodes $(2,4,5,6)$ for exact repair of node 1: 1 systematic node and 3 parity nodes. The idea is to remap one parity node to make it look like a systematic node. We then virtually connect to 2 systematic and 2 parity nodes. Specifically, we can remap node 6 with $\mathbf{c}'^t$:
\begin{align}
\mathbf{c}'^t = \mathbf{a}^t \mathbf{A}_3 +
\mathbf{b}^t \mathbf{B}_3 +
\mathbf{c}^t \mathbf{C}_3
\end{align}
We can then rewrite node 4 and 5 in terms of $\mathbf{a}$, $\mathbf{b}$ and $\mathbf{c}'$ and therefore we virtually have connection with 2 systematic and 2 parity nodes. Note that corresponding encoding submatrices after remappring are still diagonal matrices. Hence, we can apply the same procedures as those in  Section~\ref{sec:FrameworkSystematic}.
\end{remarks}




\subsection{The MDS-Code Property}
\label{sec-MDScodeProperty}

The third part of our framework is to guarantee the MDS-code property. Consider all four possibilities corresponding to the Data Collector (DC) contacting (1) 3 systematic nodes; (2) 3 parity nodes; (3) 1 systematic and 2 parity nodes; (4) 1 systematic and 2 parity nodes.

The first is a trivial case. The second case has been already verified in the process of forming the dual structure. The third case requires the invertibility of all of each encoding submatrix. In this case, it is obvious since encoding submatrix is diagonal and each element is non-zero. The last case is also easy to check.
Consider a specific example where the DC connects to nodes 3, 4 and 5. In this case, we first recover $\mathbf{c}$ from node 3 and  subtract the terms associated with $\mathbf{c}$ from nodes 4 and 5. We then get:
\begin{align}
\begin{split}
\left[
  \begin{array}{cc}
    \mathbf{a}^t &  \mathbf{b}^t\\
  \end{array}
\right]
\left[
  \begin{array}{cc}
    \mathbf{A}_1 & \mathbf{A}_2 \\
    \mathbf{B}_1 & \mathbf{B}_2 \\
  \end{array}
\right].
\end{split}
\end{align}
Again, using the technique in~\cite{Jafar:IC} and Schwartz-Zippel lemma, for a sufficiently large field size, this composite matrix has non-zero determinant with probability 1.

\subsection{Existence of Codes}
As mentioned several times, for a sufficiently large field size, a random construction for encoding submatrices suffices to guarantee exact repair of all nodes and MDS-code property with probability 1. Hence, we obtain the following theorem.

\begin{lemma}[$(6,3,4)$ Exact-Repair MDS Codes]
\label{theorem-63}
There exist vector linear Exact-Repair MDS codes that achieve the minimum repair bandwidth corresponding to the cutset bound of~(\ref{eq-MSRpoint}), allowing for any failed node to be exactly repaired with access to any arbitrary $d=4$ survivor nodes, provided storage symbols can be split into a sufficiently large number of subsymbols, and the field size can be made sufficiently large.
\end{lemma}

\section{Generalization}
\label{sec-generalization}

As one can easily see, the interference alignment technique described in Section~\ref{sec:FrameworkSystematic} can be generalized to all admissible values of $(n,k,d)$, i.e., $k <n$ and $k \leq d \leq n-1$.

\begin{theorem}[$(n,k,d)$ Exact-Repair MDS Codes]
\label{theorem-any(nkd)}
  There exist vector linear Exact-Repair MDS codes that achieve the minimum repair bandwidth corresponding to the cutset bound of~(\ref{eq-MSRpoint}), allowing for any failed node to be exactly repaired with access to any arbitrary $d$ survivor nodes, where $k \leq d \leq n-1$, provided storage symbols can be split into a sufficiently large number of subsymbols, and the field size can be made sufficiently large.
   \end{theorem}
\begin{proof}
In the interests of conceptual simplicity, and to parallel the analysis of the $(6,4,3)$ example described earlier, we provide only a sketch of the proof for the general case. This can be formalized to be precise at the cost of much heavier notational clutter, which we consciously avoid.

\textbf{Systematic Node Repair:} Let $\mathbf{G}_{l}^{(i)}$ indicate an encoding submatrix for parity node $i$, associated with information unit $l$, where $1 \leq i \leq n-k$ and $1\leq l \leq k$. Let $\mathbf{w}_l$ be $l$th information-unit vector. Without loss of generality, consider exact repair of systematic node 1.
Using vector linear codes, we split each symbol into $B=m^N$ number of subsymbols, where $m$ is an arbitrarily large positive integer and the exponent $N$ is given by
\begin{align}
N = (k-1) (d-k+1).
\end{align}
The maximum file size (based on the cutset bound of~(\ref{eq-MSRpoint})) is $\mathcal{M}=k (d-k+1)$ units, inducing a storage cost $\alpha=d-k+1$ units. Since each subsymbol has $\frac{1}{m^N}$ capacity, each storage contains $\alpha m^N (= (d-k+1) m^N)$ number of subsymbols. Note that the size of encoding submatrices is $\alpha m^N$-by-$ \alpha m^4$.

A failed node 1 is exactly repaired through the following steps. Suppose without loss of generality that we connect $k-1$ systematic nodes and first $d-k+1$ parity nodes\footnote{As mentioned earlier, we can convert the other connection configurations into this particular configuration with the remapping technique}. Parity survivor nodes project their data using the following projection matrix:
\begin{align}
\mathbf{V}:=[\mathbf{v}_1, \cdots, \mathbf{v}_{m^N}] \in \mathbb{F}_q^{ \alpha m^N \times m^N},
\end{align}
where $\mathbf{v}_i \in \mathcal{V}$. The set $\mathcal{V}$ is defined as:
\begin{align}
\label{eq:VformG}
\mathcal{V} := \left \{ \prod_{i=1,\cdots, d-k+1, l=2,\cdots,k} \left[ \mathbf{G}_l^{(i)} \right]^{e_{i,l}}
 \mathbf{w}: e_{i,l} \in \left \{ 1, \cdots , m \right \} \right \},
\end{align}
where $\mathbf{w}=[1,\cdots,1]^t$. Note that $|\mathcal{V}| \leq m^N$.

Let us consider the equations downloaded from parity nodes:
\begin{align}
\begin{split}
&\mathbf{w}_1^t (\mathbf{G}_1^{(1)} \mathbf{V}) +
\mathbf{w}_2^t (\mathbf{G}_2^{(1)} \mathbf{V}) + \cdots +
\mathbf{w}_k^t (\mathbf{G}_k^{(1)} \mathbf{V}); \\
&\qquad \qquad \vdots \\
&\mathbf{w}_1^t (\mathbf{G}_1^{(d-k+1)} \mathbf{V}) +
\mathbf{w}_2^t (\mathbf{G}_2^{(d-k+1)} \mathbf{V}) + \cdots +
\mathbf{w}_k^t (\mathbf{G}_k^{(d-k+1)} \mathbf{V}).
\end{split}
\end{align}
Note that by (\ref{eq:VformG}), for $l \neq 1$, any column vector in $[\mathbf{G}_l^{(1)}\mathbf{V} , \cdots, \mathbf{G}_l^{(d-k+1)}\mathbf{V}]$ is an element of $\mathcal{\bar{V}}$ defined as:
\begin{align}
\label{eq:VbarformG}
\mathcal{\bar{V}} := \left \{ \prod_{i=1,\cdots, d-k+1, l=2,\cdots,k} \left[ \mathbf{G}_l^{(i)} \right]^{e_{i,l}}
 \mathbf{w}: e_{i,l} \in \left \{ 1, \cdots , m+1 \right \} \right \},
\end{align}
This implies that for $l \neq 1$, ${\sf rank}[\mathbf{G}_l^{(1)}\mathbf{V} , \cdots, \mathbf{G}_l^{(d-k+1)}\mathbf{V}] \leq (m+1)^N$. This allows for simultaneous interference alignment. Systematic survivor nodes project their data using the following \emph{projection matrix}:
\begin{align}
\mathbf{\bar{V}}:=[\mathbf{\bar{v}}_1, \cdots, \mathbf{\bar{v}}_{(m+1)^N}] \in \mathbb{F}_q^{\alpha m^N \times (m+1)^N},
\end{align}
where $\mathbf{\bar{v}}_i \in \mathcal{\bar{V}}$. We can then guarantee that for $l \neq 1$:
\begin{align}
\begin{split}
&{\sf span}[\mathbf{G}_l^{(1)}\mathbf{V} , \cdots, \mathbf{G}_l^{(d-k+1)}\mathbf{V}] \subset {\sf span [\mathbf{\bar{V}}}].
\end{split}
\end{align}
Hence, we can clean out any interference. Now let us consider the decodability of desired signals. To successfully recover $\mathbf{w}_1$, we need:
\begin{align}
\label{eq:fullrank_desiredG}
{\sf rank}[\mathbf{G}_1^{(1)}\mathbf{V} , \cdots, \mathbf{G}_1^{(d-k+1)}\mathbf{V}] = (d-k+1) m^N = \alpha m^N.
\end{align}
Using the technique in~\cite{Jafar:IC} and Schwartz-Zippel lemma~\cite{Motwani}, we can ensure the (\ref{eq:fullrank_desired}) \emph{with probability 1} for a sufficiently large field size.

Finally we validate that total repair bandwidth is:
\begin{align}
\begin{split}
\gamma &= (k-1) \frac{(m+1)^N}{m^N} + (d-k+1) \cdot \frac{m^N}{m^N} \\
&\longrightarrow d.
\end{split}
\end{align}
Note that as $m$ goes to infinity, total repair bandwidth approaches minimum repair bandwidth matching the cutset lower bound of~(\ref{eq-MSRpoint}).

\textbf{Parity Node Repair:} As discussed in Section~\ref{sec:ParityNodeRepair}, we can draw a dual structure by remapping parity nodes with primed new notations. The key observation is that newly mapped encoding submatrices are still diagonal matrices. Hence, we can apply the same technique used in systematic node repair.

\textbf{MDS-Code Property:}
We check the invertibility of a composite matrix when a Data Collector connects to $i$ systematic nodes and $k-i$ parity nodes for $i=0,\cdots, k$. As mentioned earlier, for a sufficiently large field size, the composite matrix has non-zero determinant with probability 1.
\end{proof}

\section{Conclusion}
Using interference alignment techniques, we have shown the existence of vector linear Exact-Repair MDS codes that attain the cutset lower bound on repair bandwidth for all admissible values of $(n,k,d)$.
We make use of the interference alignment scheme introduced in the context of wireless interference channels in~\cite{Jafar:IC} to provide insights into Exact-Repair MDS codes. Connecting the two problems allows us to show the existence of vector linear optimal Exact-Repair MDS codes in distributed storage systems.

\bibliographystyle{IEEEtran}
\bibliography{Storage_IA}

\begin{thebibliography}{10}
\providecommand{\url}[1]{#1}
\csname url@samestyle\endcsname
\providecommand{\newblock}{\relax}
\providecommand{\bibinfo}[2]{#2}
\providecommand{\BIBentrySTDinterwordspacing}{\spaceskip=0pt\relax}
\providecommand{\BIBentryALTinterwordstretchfactor}{4}
\providecommand{\BIBentryALTinterwordspacing}{\spaceskip=\fontdimen2\font plus
\BIBentryALTinterwordstretchfactor\fontdimen3\font minus
  \fontdimen4\font\relax}
\providecommand{\BIBforeignlanguage}[2]{{%
\expandafter\ifx\csname l@#1\endcsname\relax
\typeout{** WARNING: IEEEtran.bst: No hyphenation pattern has been}%
\typeout{** loaded for the language `#1'. Using the pattern for}%
\typeout{** the default language instead.}%
\else
\language=\csname l@#1\endcsname
\fi
#2}}
\providecommand{\BIBdecl}{\relax}
\BIBdecl

\bibitem{Dimakis:INFOCOM}
A.~G. Dimakis, P.~B. Godfrey, Y.~Wu, M.~Wainwright, and K.~Ramchandran,
  ``Network coding for distributed storage systems,'' \emph{IEEE INFOCOM},
  2007.

\bibitem{Wu:Allerton}
Y.~Wu, A.~G. Dimakis, and K.~Ramchandran, ``Deterministic regenerating codes
  for distributed storage,'' \emph{Allerton Conference on Control, Computing
  and Communication}, Sep. 2007.

\bibitem{Sameer:ISIT2010}
S.~Pawar, S.~E. Rouayheb, and K.~Ramchandran, ``On secure distributed data
  storage under repair dynamics,'' \emph{to appear in IEEE ISIT}, 2010.

\bibitem{KumarRamchandran_MSR}
N.~B. Shah, K.~V. Rashmi, P.~V. Kumar, and K.~Ramchandran, ``Explicit codes
  minimizing repair bandwidth for distributed storage,'' \emph{IEEE ITW, Jan.
  2010, online avaiable at arXiv:0908.2984v2}, Sep. 2009.

\bibitem{SuhR_MSR}
C.~Suh and K.~Ramchandran, ``Exact regeneration codes for distributed storage
  repair using interference alignment,'' \emph{to appear in IEEE ISIT, June
  2010, online avaiable at arXiv:1001.0107v2}, Apr. 2010.

\bibitem{Mohammad}
M.~A. Maddah-Ali, S.~A. Motahari, and A.~K. Khandani, ``Communication over
  {MIMO} {X} channels: Interference alignment, decomposition, and performance
  analysis,'' \emph{IEEE Transactions on Information Theory}, vol.~54, pp.
  3457--3470, Aug. 2008.

\bibitem{Jafar:IC}
V.~R. Cadambe and S.~A. Jafar, ``Interference alignment and the degree of
  freedom for the {K} user interference channel,'' \emph{IEEE Transactions on
  Information Theory}, vol.~54, no.~8, pp. 3425--3441, Aug. 2008.

\bibitem{Koetter:it}
R.~Koetter and M.~Medard, ``An algebraic approach to network coding,''
  \emph{IEEE/ACM Transactions on Networking}, vol.~11, no.~5, Oct. 2003.

\bibitem{ahlswede:it}
R.~Ahlswede, N.~Cai, S.-Y.~R. Li, and R.~W. Yeung, ``Network information
  flow,'' \emph{IEEE Transactions on Information Theory}, vol.~46, no.~4, pp.
  1204--1216, Jul. 2000.

\bibitem{Wu:ISIT}
Y.~Wu and A.~G. Dimakis, ``Reducing repair traffic for erasure coding-based
  storage via interference alignment,'' \emph{Proc. of IEEE ISIT}, 2009.

\bibitem{Motwani}
R.~Motwani and P.~Raghavan, \emph{Randomized Algorithms}.\hskip 1em plus 0.5em
  minus 0.4em\relax Cambridge University Press, 1995.

\end{thebibliography}

\end{document}